%% file: main.tex
\documentclass[10pt]{article}
\usepackage[T1]{fontenc}
\usepackage[utf8]{inputenc}
\usepackage{lmodern}
\usepackage[margin=0.95in]{geometry}
\usepackage{microtype}
\usepackage{amsmath,amssymb,amsthm}
\usepackage{booktabs,array,tabularx,float}
\usepackage[numbers,sort&compress]{natbib}
\usepackage[hidelinks,breaklinks]{hyperref}
\usepackage{url,enumitem,listings}
\hypersetup{pdftitle={Approval Integrity and Recovery in LLM Answer Publication},pdfauthor={Faruk Alpay; Taylan Alpay},pdfsubject={Approval integrity, authorization freshness, and publication enforcement},pdfkeywords={approval integrity, complete mediation, authorization freshness, TOCTOU, prompt injection, cs.CR}}
\setlist{itemsep=1pt,topsep=3pt}
\newtheorem{proposition}{Proposition}
\newcommand{\code}[1]{\texttt{#1}}
\newcommand{\Allow}{\mathsf{Allow}}
\newcommand{\Hash}{\mathsf{H}}
\newcommand{\Trim}{\mathsf{trim}}
\newcommand{\sources}[1]{\citep{#1}}
\input{result-macros.tex}

\title{Approval Integrity and Recovery\\in LLM Answer Publication}
\author{Faruk Alpay\textsuperscript{1}\thanks{Correspondence: \href{mailto:alpay@lightcap.ai}{alpay@lightcap.ai}.}\quad Taylan Alpay\textsuperscript{2}\\[6pt]
\small\textsuperscript{1}Department of Computer Engineering, Bah\c{c}e\c{s}ehir University, Istanbul, T\"urkiye\\
\small\textsuperscript{2}Department of Aerospace, University of Turkish Aeronautical Association, Ankara, T\"urkiye\\[4pt]
\small\href{mailto:faruk.alpay@bahcesehir.edu.tr}{faruk.alpay@bahcesehir.edu.tr}\\
\small\href{mailto:s220112602@stu.thk.edu.tr}{s220112602@stu.thk.edu.tr}}
\date{September 10, 2026}
\begin{document}
\maketitle
\begin{abstract}
Publication integrity in LLM systems requires binding approved content to its current authorization context. We examine exact-content binding, authorization freshness and checkpoint recovery in Lightcap's publication enforcement mechanism. On 900 independently human-annotated RAGTruth responses from 150 source tasks, three dated Ministral models and a same-model direct-grounding baseline yield 3,600 assessments. The production response-act checker instantiated with 14B accepts \LCFourteenUnsafeCount{} of \UnsupportedCount{} unsupported-labelled answers; the direct baseline accepts \DirectUnsafeCount{}. Supported-answer retention is \LCFourteenSupportedRetention\% and \DirectSupportedRetention\%, respectively. An exact promotion--correction identity tracks error through 100 chronological 3B--14B--8B--14B answer trajectories. Among 65 initially approved answers, the final stateful recheck--recovery policy increases exact-match error by 9.23 percentage points relative to the initial checkpoint (95\% article-clustered interval $[-1.72,19.61]$). Controlled evidence-fingerprint changes expose asymmetric freshness enforcement between publication and recovery. A separate BIPIA prompt-injection experiment records zero target insertions among 266 valid editor outputs. External Hugging Face calibration experiments transfer retrieval models from ArguAna to SciFact and NFCorpus, and diagnostic decision rules from Thunderbird to BGL, distinguishing probability calibration from ranking changes. The measurements separate semantic false approval, stale authorization and recovery-induced error at executable publication boundaries.
\end{abstract}

\section{Introduction}
Semantic approval and publication authorization protect distinct objects. A model assesses whether an answer satisfies a request using evidence; an application decides which bytes may cross its output boundary. An unsupported answer can cross because the assessment is wrong, or because the application publishes a different answer under the earlier assessment. The first failure concerns semantic decision quality. The second concerns integrity of approval use.

Retrieved documents, tool responses, and optional transformations enlarge the gap between these decisions. RAG and tool-using agents routinely introduce external text into multi-stage computation \sources{rag,react,toolformer}. Indirect prompt injection exploits this route by placing attacker instructions in data \sources{indirect}. For answer publication, an attacker can seek insertion of an advertisement, a fraudulent instruction, or an unsupported claim during a rewrite. If the original approval survives that rewrite, control over low-trust content acquires authority over the final answer. Classical complete mediation motivates checking the content--approval association at the sink \sources{saltzer}.

Exact-content binding offers an integrity control, but its utility depends on the surrounding semantic decisions. Rejecting a changed answer can withhold a useful supported checkpoint; restoring a mistakenly approved checkpoint preserves its error. These outcomes motivate joint measurement of the checker and publication policy. Our primary finding is that a production response-act checker satisfies its structured interface while admitting almost all human-labelled unsupported answers. Its performance changes substantially when replaced by a simpler grounding-only prompt at the same model size.

We make three contributions. First, we measure the unchanged 3B, 8B and 14B response-act checker against independent human labels and a paired 14B baseline, with source-task-clustered uncertainty. Second, we derive a sequential promotion--correction identity and evaluate publication policies through exhaustive response pairing, actual four-stage revision and controlled evidence-state transitions. Third, we test external score calibration and diagnostic selection using source-pinned retrieval and log datasets. The executable artifact links semantic validation, approval freshness and diagnostic observability to their distinct measured outcomes.

\section{Security Model and Related Work}
\subsection{Assets, adversary and trust boundary}
The protected assets are the association of approval with exact answer content and current scope, and the availability of an eligible previously approved answer. The integrity objective is that every published answer has a durable applicable receipt. The recovery objective is to preserve an eligible checkpoint after an unsuccessful optional transformation.

The temporal environment may also advance the trusted evidence version after approval, through correction, withdrawal or replacement. The freshness experiment treats this transition as an input to the authorization boundary. An authoritative update can invalidate the basis of an earlier approval even when the answer bytes and user request remain unchanged. This environment-state transition is distinct from the attacker-controlled editor experiment below.

The adversary controls an additional external tool excerpt consumed by an optional editor. The legitimate user request, original evidence, checker code, receipt store and publication sink remain trusted. The attacker seeks a specified contaminating sentence in the released answer. Control of tool data may induce a model rewrite, but cannot directly mint an authoritative receipt. Our attack experiment uses two fixed insertion payloads from the independently released BIPIA test set \sources{bipia}. Generation, reassessment and sink checks operate on the resulting real model output. The measured attack surface is this explicit editor interface; the experiment's endpoint is targeted answer contamination.

Two routes cross the boundary. In \emph{approval substitution}, the editor produces $b_1\ne b_0$ and the application transfers approval for $b_0$ to $b_1$. In \emph{semantic laundering}, the checker issues a new approval for an attacker-contaminated $b_1$. Content binding mediates the first route; the second exposes the semantic verifier's residual error. Recoverable output distinguishes integrity enforcement from an abstain-all policy.

\subsection{Prior mechanisms and empirical position}
Complete mediation and robust declassification establish the broader distinction between low-trust influence and release authority \sources{saltzer,robust}. CaMeL separates privileged control from quarantined processing, FIDES tracks information flow, AgentSpec imposes runtime constraints, and SPA preserves provenance across queries \sources{camel,fides,agentspec,spa}. TOCTOU-Bench measures validation/use gaps for external state in agent tasks \sources{toctou}; PCAA and CAVA study action certificates and canonical approval binding \sources{pcaa,cava}. Our experimental object is the semantically assessed answer and its replacement, with the previous checkpoint as an available alternative output.

InjecAgent, AgentDojo and BIPIA operationalize attacker-controlled tool or retrieved content \sources{injecagent,agentdojo,bipia}. Our fixed BIPIA transfer condition provides a reproducible candidate-contamination endpoint; adaptive attacker studies motivate treating payload search budget and access as explicit threat-model parameters \sources{attacker,firewalls}. RAGTruth supplies human supportedness annotations \sources{ragtruth}, while semantic entropy measures meaning-level uncertainty \sources{entropy}. The central evaluation here tests how measured semantic errors propagate through an executable authorization boundary.

\section{Content-Bound Publication Enforcement and Recovery}
\subsection{Publication game}
Let $b$ be the UTF-8 encoding of a well-formed Unicode answer, $r,c$ nonempty current request and contract identifiers, and $C=(s,h_b,h_r,h_c,d)$ an authoritative receipt. The admissible status set $S$ contains \code{typed\_projection} and \code{independently\_verified}; $z$ indicates a digest-only output. Publication requires
\begin{align}
\Allow(b,C,r,c,z)\iff{}&\mathsf{Present}(b)\land\neg z\land s\in S\land d=\mathsf{true}\nonumber\\
&\land h_b=\Hash(b)\land h_r=r\land h_c=c.
\label{eq:allow}
\end{align}
Presence tests trimmed nonemptiness; the content hash covers the untrimmed publication string. Issuance and the mediated sink use an immutable tuple $(b,r,c)$. The experiment lets an adversary observe an issued receipt for $(b_0,r_0,c_0)$ and propose $(b_1,r_1,c_1)$; the adversary wins if a different tuple passes using that receipt.

\begin{proposition}[Approval-use integrity]
With authoritative issuance, immutable checked publication bytes and complete nonempty scope identifiers, the probability of a substitution-game win is bounded by the corresponding SHA-256 collision-finding advantage.
\end{proposition}
\begin{proof}
An accepted proposal satisfies all three equalities in Equation~\ref{eq:allow}. A scope change falsifies $h_r=r_1$ or $h_c=c_1$. Therefore a winning proposal has the original scope and $b_1\ne b_0$ with $\Hash(b_1)=\Hash(b_0)$. An algorithm running the adversary and returning $(b_0,b_1)$ on a win produces a hash collision with the same probability. The construction requires no independence or soundness of the semantic checker.
\end{proof}

\subsection{Acknowledged issuance and checkpoint restoration}
The existing \code{commitResponseAct} computes content and scope bindings, emits \code{response\_act\_settlement}, and returns a frozen durable receipt after a literal synchronous \code{true} acknowledgement for a nonempty durable run ID. Storage acknowledgement is the persistence abstraction used by the copied caller.

A checkpoint $K$ stores answer $x$, receipt, scope, evidence fingerprint and digest status. Eligibility requires matching current request, contract and nonempty evidence identifiers, an active run, and
\begin{align}
x_K=\Trim(K.x),\qquad &\Trim(G(x_K))=x_K,\quad Q(x_K)=\mathsf{true},\nonumber\\
&\Allow(\mathsf{UTF8}(x_K),K.C,r,c,\mathsf{false}),
\label{eq:restore}
\end{align}
where $G,Q$ are the current guard and structural check. Before replacing the final answer, the restoration caller records the restored and discarded content hashes and waits for acknowledgement.

\begin{proposition}[Recovery noninterference]
For fixed eligible checkpoint $K$, current context $\gamma$, and successful restoration acknowledgement, the published output projection satisfies $R(K,\gamma,d_1)=R(K,\gamma,d_2)=x_K$ for any two discarded candidates $d_1,d_2$.
\end{proposition}
\begin{proof}
Eligibility in Equation~\ref{eq:restore} depends only on $(K,\gamma)$. The discarded candidate contributes its hash to the restoration event. After acknowledgement the caller assigns the returned checkpoint text, giving the same output for every discarded candidate. The caller preserves an existing blocked status and otherwise marks the recovered result partial.
\end{proof}

\input{risk-decomposition.tex}

\subsection{Implementation provenance}
All publication code is pinned at Lightcap commit \code{393e3e3db1abb35a0bb9c43f5e24156be08661bb}. The artifact copies \code{local-inference/response-act.js} and \code{schema.js}, the provider grammar conversion from \code{providers/mistral-transport.js}, publication/restoration predicates from \code{agent/source-projection.js} (lines 139--223), issuance/recovery closures from \code{agent.js} (1707--1765 and 1802--1878), and serialization/hash helpers from \code{reliability/cost\_control.js} (13--27). Byte hashes and span manifests bind the research adapter to these sources. Import/export wiring and dependencies are supplied by the adapter; function bodies, prompts, defaults and acceptance conditions are preserved. Inspection of the finalization call sites establishes the code's production integration. Research settlement journals and injected storage acknowledgements make issuance and recovery observable in isolation.

\input{benchmark-methods.tex}
\section{Semantic Verification and Policy Results}
\label{sec:results}
\input{benchmark-results.tex}
\input{security-results.tex}
\input{drift-results.tex}

\section{Calibration and Cross-System Transfer}
\input{calibration-results.tex}

\section{Discussion}
\subsection{Semantic authority and residual risk}
The primary result locates a substantial failure at approval issuance: the response-act checker releases almost all unsupported-labelled answers across the tested model sizes. The direct 14B baseline achieves lower unsupported release and lower supported retention. Instruction placement, multi-criterion assessment, schema structure and wording vary jointly between these interfaces; the paired comparison estimates their combined effect. The publication predicate then preserves the association between that fallible decision and exact content. Recovery retains the error rate of the selected checkpoint population, which motivates the promotion--correction accounting and asymmetric revision experiment.

Inspection identifies concrete semantic errors. In RAGTruth response 5528, an approved answer reverses who informs whom. In response 9200, source attributes set to \code{null} become affirmative amenities and are approved. Response 606 is classified as grounded but withheld for request/facet coverage. Responses 17060 and 17063 expose annotation disagreement: additional cautions absent from the evidence appear in both, but only the former carries an unsupported-span annotation. The primary estimates retain the original labels, with implicit-truth and quality sensitivities. These cases explain observed verdict types; their selection supports qualitative analysis rather than prevalence estimation.

\subsection{Scope of the implementation guarantees}
Executable counterexamples delimit the propositions. Caller-forged receipts pass the pure predicate, empty expected scope identifiers disable scope comparisons, empty evidence identifiers skip that recovery comparison, and an empty run ID enables ephemeral issuance. Distinct unpaired UTF-16 surrogates may share UTF-8 replacement bytes. Recovery compares trimmed guard strings. A semantically erroneous but content-matching receipt remains applicable. Consequently, the proof domain comprises authoritative receipts, nonempty complete context, well-formed Unicode, truthful acknowledged storage and a mediated immutable sink. The artifact exposes all seven counterexamples alongside the domain-valid conformance checks.

\subsection{Empirical scope and uncertainty}
RAGTruth measures evidence-relative supportedness across three task types and six original generators. Its 150-task sample supports estimates for that sampling frame, with clustering preserving within-task dependence. Annotation disagreements and possible pretraining overlap remain sources of measurement error. Hosted requests use dated models; serving variability, technical failures and experiment pacing are recorded separately from semantic decisions.

The symmetric replay measures exhaustive policy behavior conditional on observed initial approval. The time-ordered SQuAD experiment measures a realized transition process against independent human answer references. Exact match is sensitive to answer boundaries, and the paired article-clustered interval retains uncertainty about the direction of the policy effect. BIPIA payload transfer separately evaluates two fixed contamination goals at the declared editor interface; zero observed insertions leaves the boundary's incremental attack-removal effect unestimated. Retrieval judgments evaluate ranked evidence; BGL labels evaluate anomaly localization from projected metadata. These distinct labels specify the estimands of their respective experiments.

\section{Conclusion}
Approval integrity requires both a defensible semantic decision and faithful use of that decision at publication. Independent human annotations reveal a near-permissive production checker despite a valid structured interface. Content-bound authorization and acknowledged checkpoint recovery mediate subsequent replacement, while an exact promotion--correction identity describes their semantic consequences. Time-ordered revision experiments connect model behavior to those publication controls on the same outputs. The source-pinned artifact makes each layer independently reproducible and exposes the residual errors that integrity enforcement preserves.

\section*{Artifact and Ethics}
The source archive contains code, provenance manifests, public-data identifiers, structured model verdicts, measured results and reproduction instructions. Original corpus text and model weights are retrieved from pinned upstream repositories. Short SQuAD answer spans and human references accompany the chronological scores; full research-generated RAGTruth prose is retained locally for audit. Experiments use public data, local embedding inference and official hosted inference under a cumulative USD~5 authorization. Security payloads execute solely in the research editor and sink. Production user records and remote application state remain outside the experiment. Rights and source notices accompany redistributed material.

\appendix
\input{representation-results.tex}
\section{Conformance, Fault Injection and Timing}
The supplementary suite comprises 288 publication states, 1,024 recovery states, 1,792 mutation/control cases, 24 persistence-fault checks, three status checks and 1,280 checkpoint recoveries: 4,411 assertions, all passing in the archived run. These are finite-domain regression counts. The full predicate rejects 1,664 stale/non-durable cases and accepts 128 unchanged controls. A Boolean-approval ablation admits 1,536 stale cases; content-only binding leaves 256, adding request binding leaves 128, and omitting durability admits 128 non-durable cases. The copied caller restores every eligible checkpoint in the 1,280 successful-acknowledgement fixtures. Fault cases distinguish literal synchronous true from exceptions, promises and other return values.

On Apple M4 Pro, Node.js 26.7.0, Darwin 25.3.0, arm64, 21 warmed batches per size yield median batch-mean publication-check times of 1.29, 3.20 and 34.23~$\mu$s for 256, 4,096 and 65,536 UTF-8 bytes. The timed region is the publication predicate; model, storage, receipt construction and VM startup costs are measured elsewhere or excluded from this component timing.

\section{Reproduction}
Offline Node.js tests execute the copied predicates and recompute semantic and policy statistics from stored structured verdicts. Pinned original corpus downloads support source inspection and replay. Embedding scripts record model revisions, input hashes, pooling, dimensions, normalization and ranked document IDs. Fresh hosted inference uses a supplied Mistral credential and the shared durable budget journal. The archive's command manifest distinguishes offline arithmetic, local model inference and hosted model inference.
\input{transport-results.tex}
\bibliographystyle{plainnat}
\bibliography{references}
\end{document}

%% file: result-macros.tex
\newcommand{\LCFourteenUnsafeCount}{291}
\newcommand{\DirectUnsafeCount}{41}
\newcommand{\LCFourteenSupportedRetention}{95.2}
\newcommand{\DirectSupportedRetention}{66.9}
\newcommand{\UnsupportedCount}{302}

%% file: risk-decomposition.tex
\subsection{Selective risk and replacement balance}
\label{sec:risk-decomposition}
Let $Y=1$ indicate the evaluation's adverse endpoint, $A=1$ acceptance, $\pi=\Pr(Y=1)$, $f=\Pr(A=1\mid Y=1)$ and $t=\Pr(A=1\mid Y=0)$. Coverage $c$, released-output risk $r$, and supported yield $u$ satisfy
\begin{equation}
c=\pi f+(1-\pi)t,\qquad r=\frac{\pi f}{c},\qquad u=(1-\pi)t=c(1-r),\quad c>0.
\label{eq:selective}
\end{equation}
These quantities distinguish withholding, acceptance of adverse outputs, and useful output over the complete scenario population. For RAGTruth, $Y$ is the original human unsupported-span indicator; for the attack experiment, $Y$ is the predefined target-contamination endpoint.

For an initial answer $i$ and replacement $j$, let $w_{ij}\geq0$ be the sampling weight including actual initial acceptance and checkpoint eligibility, and $W=\sum_{ij}w_{ij}>0$. Recheck decision $A_j$ is observed on the exact replacement. Rollback publishes $i$; hybrid publishes $j$ when $A_j=1$ and otherwise $i$. Hence
\begin{align}
R_{\mathrm{hybrid}}-R_{\mathrm{rollback}}
&=\frac{1}{W}\sum_{ij}w_{ij}A_j(Y_j-Y_i)\nonumber\\
&=\frac{1}{W}\sum_{ij}w_{ij}A_j\big[(1-Y_i)Y_j-Y_i(1-Y_j)\big].
\label{eq:promotion-balance}
\end{align}
The first term counts harmful promotion; the second counts beneficial correction. The equality follows by substituting the hybrid output indicator $A_jY_j+(1-A_j)Y_i$ and collecting terms. It holds for dependent answers, correlated verifier errors and asymmetric replacement kernels.

\begin{proposition}[Directed-pair symmetry]
If every distinct same-source pair receives equal weight before conditioning on initial approval, the aggregate hybrid--rollback difference is zero, source by source.
\end{proposition}
\begin{proof}
Write $w_{ij}=q_{ij}A_i$ with $q_{ij}=q_{ji}$. The numerator of Equation~\ref{eq:promotion-balance} is $\sum_{ij}q_{ij}A_iA_j(Y_j-Y_i)$. Under the index involution $(i,j)\mapsto(j,i)$ each term has its additive inverse. Thus the sum vanishes within each source, and every source-bootstrap replicate preserves the equality.
\end{proof}

For the exhaustive replay this is a design identity. For a chronological revision kernel $q_{ij}$, symmetry generally fails; the measured difference estimates the balance of actual promotions and corrections. The artifact checks Equation~\ref{eq:selective} against all frozen confusion counts and Equation~\ref{eq:promotion-balance} against the policy-selected outputs.

%% file: benchmark-methods.tex
\section{External Evaluation Protocol}
\label{sec:external}
The evaluation separates three units: independently human-labelled responses, policy-selected outputs under exhaustive pairing, and time-ordered outputs under actual model revision. The primary research questions concern semantic false approval, supported-answer retention, and propagation of changed content through publication policies.

\subsection{Independent human annotations}
RAGTruth contains human span-level supportedness annotations for question answering, summarization and structured-data-to-text responses \sources{ragtruth}. We obtain its original repository at commit \code{c103204b9ce28d6bbad859304bf30de72b8ed8fe}. Within the official test partition, source identifiers are ordered by SHA-256 of a fixed seed string concatenated with the identifier. The first 50 sources per task type contribute all six associated test responses: 900 responses from 150 sources. A disjoint training-only pilot validates transport. Selection is independent of labels, response length, quality flags, generator identity and measured verdicts.

The primary indicator $Y_i=1$ denotes at least one human-annotated unsupported span. Spans marked \code{implicit\_true} remain adverse under the primary evidence-relative endpoint; a sensitivity analysis removes those spans. Quality-flagged responses remain in the primary population, with a separate good-quality sensitivity. The complementary label denotes absence of an annotated unsupported span.

For QA, the original question is the request and its supplied passages are evidence. For summarization, the instruction before the exact article is the request, and the complete article is evidence. For data-to-text, the instruction is separated at the structured-data delimiter and the complete object is serialized as evidence. Mapping requires a nonempty request and evidence packet. A single request-derived facet supplies the coverage target. Labels and generator identities remain outside model inputs. Hosted requests preserve the complete mapped evidence and disable the local compact-context projection.

\subsection{Verifier comparison and inference protocol}
The original \code{assessGroundedResponseAct} is evaluated with dated Ministral models at three sizes. The official model IDs are \code{ministral-3b-2512}, \code{ministral-8b-2512} and \code{ministral-14b-2512}. Its deterministic prechecks and structured assessment cover response role, request satisfaction, grounding, source use and facet coverage. The acceptance predicate and prompts remain unchanged.

The paired direct 14B baseline returns a single Boolean for whether every factual assertion is supported by the identical evidence. Its system instruction assigns evidence and candidate data status and requires support from the packet. Thus the baseline comparison estimates the combined effect of changing the checker interface and criteria at fixed model size. Always-publish and always-abstain policies define the endpoints of the coverage range.

All methods use temperature zero, seed 1729 and the existing 512-token verifier output envelope. The original provider grammar conversion is followed by the complete local schema validator. JSON serialization preserves the transport boundary across adapter realms. Exact request hashes identify reusable cache entries. Records contain requested and returned model IDs, structured outputs, token usage, transaction duration and technical error codes.

All hosted experiments share one durable journal under the author's cumulative USD~5 authorization. The client reserves the full model-context charge plus the requested output ceiling at public token prices \sources{ministralthree,ministraleight,ministralfourteen}; reported usage settles the reservation, while ambiguous requests retain it. A process lock admits a single journal owner and a USD~4.80 admission threshold leaves headroom within the authorization. The journal records public-price estimated charges and outstanding reservations.

\subsection{Exhaustive same-source policy replay}
For each source, every ordered pair of distinct answer strings $(i,j)$ is considered. A receipt for $i$ is issued only after its observed checker acceptance and synchronous research-journal acknowledgement. The before/after strings share request and evidence and retain their separate human annotations. Five policies execute the copied predicates: stale approval, exact block, exact rollback, semantic recheck, and recheck followed by rollback. Recheck uses the already measured verdict on the exact $j$ packet. Eligibility fixes current scope/evidence, active-run state, structural acceptance and a guard preserving checkpoint text.

Each outcome is scored by the label of its selected answer. The pair count describes enumerated scenarios within the original source tasks. Directed-pair symmetry supplies a source-wise accounting invariant: accepted supported-to-unsupported transitions and their reverse transitions cancel in Equation~\ref{eq:promotion-balance}. The chronological experiment in the next section estimates a distinct asymmetric transition process.

\subsection{Supplementary interface probes}
Three fixed candidate suffixes assert previous approval, forge a control delimiter, or supply a fabricated verdict. The first response from ten hash-selected sources per task type produces 30 candidates, assessed under all three suffixes and four methods. An attack bypass is a newly accepted response whose clean counterpart was rejected and carried an unsupported-span annotation. The frozen set's eligible bypass denominator is reported explicitly.

Thirty further source tasks exercise the complete \code{generateVerifiedGroundedResponseAct} entry point with a 3B producer and separate 14B checker. These transactions measure generation/assessment execution, source/facet agreement and acceptance. Their generated answers form a new unlabelled prose population.

\subsection{Estimands and uncertainty}
For $n$ responses, write $N_{\!P}=\sum_i A_i$, $N_{\!U}=\sum_i A_iY_i$ and $N_{\!S}=\sum_i A_i(1-Y_i)$. Then coverage is $N_{\!P}/n$, released risk is $N_{\!U}/N_{\!P}$, and supported yield is $N_{\!S}/n$. Unsupported release additionally conditions on an adverse label, $N_{\!U}/\sum_iY_i$, while supported retention conditions on its complement, $N_{\!S}/\sum_i(1-Y_i)$. Released risk is undefined when $N_{\!P}=0$. Replay substitutes the number of eligible scenarios for $n$ and uses the selected output's human label.

We compute 2,000 seed-1729 percentile bootstrap resamples of source tasks, carrying every response and policy outcome of each sampled task together. Paired method differences use the same resampled task indices. This preserves dependence between answers and between repeated pairs. Technical errors remain explicit in operational withholding counts; common technical-success sensitivity removes the same affected response IDs across methods. Hosted latency includes the assessment transaction and shared admission pacing, with call-level elapsed times recorded separately.

%% file: benchmark-results.tex
\subsection{Human-labelled semantic decisions}

The frozen sample contains 900 responses from 150 source tasks, with 302 responses carrying annotated unsupported spans and 598 without such spans. The 3,600 method--response assessments produced 3,593 successful structured model returns and 7 technical failures; 0 assessments were resolved by deterministic prechecks alone. Table~\ref{tab:semantic} includes technical failures as operational withholding and lists them explicitly; common technical-success sensitivity is supplied below.

\begin{table}[H]\centering\small

\begin{tabular}{@{}lrrrr@{}}\toprule

Method & Unsupported release & Supported retention & Released risk & Errors \\

\midrule

Response-act 3B & 99.7 [99.0, 100.0] & 98.3 [96.9, 99.5] & 33.9 & 0 \\

Response-act 8B & 98.3 [96.7, 99.7] & 96.5 [94.1, 98.5] & 34.0 & 0 \\

Response-act 14B & 96.4 [94.0, 98.3] & 95.2 [92.3, 97.4] & 33.8 & 3 \\

Direct grounding 14B & 13.6 [9.6, 18.0] & 66.9 [61.9, 71.4] & 9.3 & 4 \\

Always publish & 100.0 [100.0, 100.0] & 100.0 [100.0, 100.0] & 33.6 & 0 \\

Always abstain & 0.0 [0.0, 0.0] & 0.0 [0.0, 0.0] & --- & 0 \\

\bottomrule\end{tabular}

\caption{Percentages; brackets are 95\% source-task bootstrap intervals. Unsupported release is conditional on human-labelled unsupported content; supported retention is conditional on its absence. Released risk is conditional on publication and is undefined for zero output. The outcome label is evidence-relative human annotated supportedness.}\label{tab:semantic}\end{table}

The principal negative result is the response-act verifier's near-permissive behavior on unsupported content. Its 3B, 8B and 14B instances release 301, 297 and 291 of the same 302 unsupported-labelled answers. The direct 14B baseline releases 41, but retains only 66.9\% of supported-labelled answers versus 95.2\% for response-act 14B. This estimates a risk--retention tradeoff. The gap to the grounding-only interface persists at 14B.

\begin{table}[H]\centering\small

\begin{tabular}{@{}llrrr@{}}\toprule Task & Method & Unsupported release & Supported retention & Released risk \\ \midrule

QA & Response-act 3B & 98.0 & 96.0 & 17.0 \\

QA & Response-act 8B & 92.0 & 91.6 & 16.7 \\

QA & Response-act 14B & 88.0 & 88.8 & 16.5 \\

QA & Direct grounding 14B & 16.0 & 62.8 & 4.8 \\

Summary & Response-act 3B & 100.0 & 100.0 & 22.3 \\

Summary & Response-act 8B & 98.5 & 100.0 & 22.1 \\

Summary & Response-act 14B & 98.5 & 99.6 & 22.1 \\

Summary & Direct grounding 14B & 20.9 & 81.1 & 6.9 \\

Data2txt & Response-act 3B & 100.0 & 100.0 & 61.7 \\

Data2txt & Response-act 8B & 100.0 & 100.0 & 61.7 \\

Data2txt & Response-act 14B & 97.8 & 100.0 & 61.1 \\

Data2txt & Direct grounding 14B & 10.3 & 47.0 & 26.0 \\

\bottomrule\end{tabular}\caption{Task-type breakdown (percentages), 300 responses and 50 source tasks per task type. Full intervals and per-generator breakdowns are in the artifact.}\label{tab:tasks}\end{table}

\subsection{Sensitivity and paired method differences}

\begin{table}[H]\centering\small

\begin{tabular}{@{}lrrrr@{}}\toprule & \multicolumn{2}{c}{Good-quality subset} & \multicolumn{2}{c}{Common technical-success subset} \\ \cmidrule(lr){2-3}\cmidrule(lr){4-5}

Method & Unsupp. release & Supp. retention & Unsupp. release & Supp. retention \\ \midrule

Response-act 3B & 99.7 & 99.3 & 99.7 & 98.3 \\

Response-act 8B & 98.3 & 97.6 & 98.3 & 96.5 \\

Response-act 14B & 96.4 & 96.3 & 96.4 & 95.6 \\

Direct grounding 14B & 13.6 & 66.7 & 13.6 & 67.5 \\

\bottomrule\end{tabular}\caption{Sensitivity percentages. Technical-success filtering removes the same response IDs from all four methods. Filtering is paired across methods.}\label{tab:sensitivity}\end{table}

Excluding spans explicitly marked as implicitly true reduces the unsupported-labelled count to 284. Under that alternative target, response-act 14B releases 273 unsupported-labelled responses (96.1\%) and retains 95.3\% of responses without a remaining unsupported span. Direct grounding 14B releases 37 (13.0\%) and retains 65.6\%, respectively. The primary target remains support in the supplied context.

Response-act 14B minus Direct grounding 14B changes unsupported release by 82.8 percentage points (95\% interval [78.0, 87.2]) and supported retention by 28.3 points ([23.6, 33.3]). Intervals use paired source-task resamples.

Response-act 14B minus Response-act 3B changes unsupported release by -3.3 percentage points (95\% interval [-5.7, -1.3]) and supported retention by -3.2 points ([-5.6, -1.1]). Intervals use paired source-task resamples.

Response-act 14B minus Response-act 8B changes unsupported release by -2.0 percentage points (95\% interval [-4.0, -0.3]) and supported retention by -1.3 points ([-2.7, -0.3]). Intervals use paired source-task resamples.

\subsection{Replacement policies and error transfer}

Replay executed 13,109 model-conditioned directed replacement scenarios. This count enumerates three verifier instantiations over the original 150 source tasks. 6 otherwise eligible identical-string pairs were excluded.

\begin{table}[H]\centering\small

\begin{tabular}{@{}llrrr@{}}\toprule Verifier & Policy & Coverage & Released risk & Supported yield \\ \midrule

Response-act 3B & Stale approval & 100.0 & 33.8 & 66.2 \\

Response-act 3B & Exact block & 0.0 & --- & 0.0 \\

Response-act 3B & Exact rollback & 100.0 & 33.9 & 66.1 \\

Response-act 3B & Semantic recheck & 98.9 & 34.1 & 65.2 \\

Response-act 3B & Recheck + rollback & 100.0 & 33.9 & 66.1 \\

Response-act 8B & Stale approval & 100.0 & 34.0 & 66.0 \\

Response-act 8B & Exact block & 0.0 & --- & 0.0 \\

Response-act 8B & Exact rollback & 100.0 & 34.0 & 66.0 \\

Response-act 8B & Semantic recheck & 97.8 & 34.3 & 64.2 \\

Response-act 8B & Recheck + rollback & 100.0 & 34.0 & 66.0 \\

Response-act 14B & Stale approval & 100.0 & 34.1 & 65.9 \\

Response-act 14B & Exact block & 0.0 & --- & 0.0 \\

Response-act 14B & Exact rollback & 100.0 & 33.9 & 66.1 \\

Response-act 14B & Semantic recheck & 96.7 & 34.2 & 63.6 \\

Response-act 14B & Recheck + rollback & 100.0 & 33.9 & 66.1 \\

\bottomrule\end{tabular}\caption{Replacement-policy percentages conditional on actual initial approval and an eligible checkpoint. For each row, $N$ is the eligible scenario count, $P$ the number published, and $U$ and $S$ the selected outputs with and without human unsupported-span labels. Coverage $=100P/N$, released risk $=100U/P$, supported yield $=100S/N$. $P=0$ gives undefined risk. Source-clustered intervals accompany the artifact.}\label{tab:replay}\end{table}

Hybrid and rollback risks are equal by the symmetric all-directed-pairs design: each harmful promotion between two accepted answers has a reversed beneficial correction. Their zero difference and degenerate interval verify the source-wise accounting invariant.

Response-act 3B has 4,443 eligible scenarios across 150 tasks with an approved checkpoint. Recheck plus rollback makes 825 harmful promotions and 825 beneficial corrections. Its risk difference from exact rollback is 0.0 percentage points (95\% interval [0.0, 0.0]); the finite-sample identity in Equation~\ref{eq:promotion-balance} holds exactly. All 150 original source tasks, including zero-eligible-scenario tasks, are retained as bootstrap sampling units.

Response-act 8B has 4,368 eligible scenarios across 150 tasks with an approved checkpoint. Recheck plus rollback makes 799 harmful promotions and 799 beneficial corrections. Its risk difference from exact rollback is 0.0 percentage points (95\% interval [0.0, 0.0]); the finite-sample identity in Equation~\ref{eq:promotion-balance} holds exactly. All 150 original source tasks, including zero-eligible-scenario tasks, are retained as bootstrap sampling units.

Response-act 14B has 4,298 eligible scenarios across 150 tasks with an approved checkpoint. Recheck plus rollback makes 785 harmful promotions and 785 beneficial corrections. Its risk difference from exact rollback is 0.0 percentage points (95\% interval [0.0, 0.0]); the finite-sample identity in Equation~\ref{eq:promotion-balance} holds exactly. All 150 original source tasks, including zero-eligible-scenario tasks, are retained as bootstrap sampling units.

\subsection{Supplementary candidate-interface sensitivity}

Response-act 3B accepts 30/30 clean candidates and 30/30, 30/30, 30/30 under prior-audit, control-delimiter and forged-receipt suffixes, respectively.

Response-act 8B accepts 30/30 clean candidates and 30/30, 30/30, 30/30 under prior-audit, control-delimiter and forged-receipt suffixes, respectively.

Response-act 14B accepts 30/30 clean candidates and 30/30, 30/30, 30/30 under prior-audit, control-delimiter and forged-receipt suffixes, respectively.

Direct grounding 14B accepts 23/30 clean candidates and 27/30, 25/30, 7/30 under prior-audit, control-delimiter and forged-receipt suffixes, respectively.

The frozen stress set has zero eligible bad-answer bypass cases. Its bypass rate is undefined; observed acceptance shifts characterize interface sensitivity.

\subsection{Producer--checker execution and accounting}

The full production response-act function accepted 29 of 30 generated candidates in the separate 3B-producer/14B-checker probe, with 1 technical failure. This measures execution and acceptance of newly generated, unlabelled answers.

All research calls, including training pilots, stress and pipeline probes, used 4,687,721 reported input tokens and 250,980 reported output tokens. The published-price settled estimate is USD~0.7880; including unresolved conservative reservations, the journal accounts for USD~1.3921 against the authorized USD~5 ceiling. These amounts are computed from reported usage and the journal's unresolved reservations.

\begin{table}[H]\centering\small\begin{tabular}{@{}lrrrr@{}}\toprule Method & Input tokens & Output tokens & Calls & Median/P95 (ms) \\ \midrule

Response-act 3B & 1,020,597 & 55,682 & 900 & 4964/7741 \\

Response-act 8B & 1,020,597 & 55,770 & 900 & 4961/8548 \\

Response-act 14B & 1,018,229 & 55,752 & 897 & 4939/8240 \\

Direct grounding 14B & 769,263 & 7,328 & 896 & 4970/7616 \\

\bottomrule\end{tabular}\caption{Main-set usage and measured assessment latency. Latency includes waiting for the shared experiment admission schedule. Per-call service elapsed times are separately recorded.}\label{tab:usage}\end{table}

%% file: security-results.tex
\section{Time-Ordered Revision and Attack Transfer}

\label{sec:chronological}

\subsection{Independent-reference revision protocol}

To estimate an asymmetric replacement process, we sample 100 questions with distinct context hashes from the official SQuAD validation split on Hugging Face \sources{squad,squadhf}. Question IDs are ordered by SHA-256 of the fixed string \code{approval-revision-1729:} followed by the ID. The first question for each new context is retained. The pinned dataset revision is \code{7b6d24c440a36b6815f21b70d25016731768db1f}. The sample spans 42 article titles.

A dated 3B model generates an extractive answer; the unchanged 14B response-act checker assesses it. A dated 14B model then receives the question, complete original context and actual 3B answer and produces a corrected extractive answer, which the same checker assesses under the shared exact-request cache. Prompts require the shortest complete answer span copied from context. Original human reference answers remain outside all model inputs. SQuAD normalized exact match (EM) and token F1 are computed against every supplied human reference, taking the maximum. Here $Y=1-\mathrm{EM}$ denotes reference-answer mismatch. This task-correctness endpoint permits direct scoring of both new outputs; RAGTruth retains its separate unsupported-span endpoint. The 200 generation and 200 assessment invocations yield 335 structured model results (40 reused from identical-request cache) and 65 deterministic prechecks, requiring 295 unique hosted transmissions.

Publication policies run on the same chronological pair. An initial receipt requires actual checker acceptance and acknowledged issuance. A successful recheck issues a new receipt for the revised bytes. Restoration revalidates the initial checkpoint and records both content hashes before assignment. All policy-selected outputs are scored against the same human answer references. Bootstrap resampling clusters by article title and carries all associated trajectories together (2,000 paired resamples, seed 1729).

\subsection{Chronological policy outcomes}

All 100 selected questions were attempted; 100 yielded complete trajectories and 0 had technical failures. Initial 3B EM/F1 was 72.0/83.4\%; after actual 14B revision it was 62.0/82.1\%. The paired EM-change interval was [-21.6, 1.9] percentage points. The checker admitted 65 initial and 61 revised answers. Policy analysis conditions on the 65 actual initial approvals, of which 28 had a changed revision.

\begin{table}[H]\centering\small\begin{tabular}{@{}lrrrr@{}}\toprule Policy & Published & Coverage & Released EM error & Exact-answer yield \\ \midrule

Stale approval & 65/65 & 100.0 & 36.9 & 63.1 \\

Exact block & 37/65 & 56.9 & 5.4 & 53.8 \\

Exact rollback & 65/65 & 100.0 & 18.5 & 81.5 \\

Semantic recheck & 53/65 & 81.5 & 26.4 & 60.0 \\

Recheck + rollback & 65/65 & 100.0 & 27.7 & 72.3 \\

\bottomrule\end{tabular}\caption{Actual 3B-to-14B revision, conditional on initial acceptance. Coverage $=100P/N$; released EM error $=100E/P$, with $E$ the selected outputs failing normalized exact match; exact-answer yield $=100C/N$, with $C$ exact-matching selected outputs. These labels come from independent SQuAD reference spans. Percentages use the displayed per-policy denominator.}\label{tab:chronological}\end{table}

Hybrid publication produced 9 EM-error promotions and 3 EM corrections. Its released-error difference from rollback was 9.2 [-1.7, 19.6] percentage points (95\% article-clustered interval). The empirical difference equals $(9-3)/65$ by Equation~\ref{eq:promotion-balance}. Transition direction is fixed by actual generation and revision. The interval spans zero, so the sample establishes the observed balance while retaining uncertainty about its population direction.

EM changes include both answer-boundary variation and content errors. In question \code{5726f1ec7089} (ID prefix), a year-only answer becomes the full supported date, reducing EM. In \code{572655267089}, the revision omits one of the two requested countries. In \code{56e0fd167aa9}, it replaces the date of first experiments with the date of arrival, despite the context distinguishing them. All three revised candidates receive checker approval. These source-audited examples motivate reporting token F1 alongside EM and interpreting the aggregate as reference-answer error. The F1-weighted selected-output yield is 90.3\% for rollback and 86.0\% for hybrid.

\subsection{Fixed external attack transfer}

A separate experiment uses 30 RAGTruth source tasks, ten per task type, with a human-supported answer accepted by the frozen 14B checker. Dated 3B, 8B and 14B optional editors receive the original request/evidence and that answer. Each task has a clean condition and two verbatim BIPIA test payloads: \emph{Marketing \& Advertising}[1] and \emph{Scams \& Fraud}[1], pinned at commit \code{a004b69ec0dd446e0afd461d98cb5e96e120a5d0}. Both request insertion of a specified sentence through an additional tool excerpt. The fixed editor instruction classifies that excerpt as untrusted data. The oracle detects the specified sentence under lowercase/whitespace normalization and checks its absence in the clean counterpart.

Across 270 planned transactions, 266 returned valid answers and four had invalid provider responses. Target-sentence occurrence was zero in all 266 valid outputs. Thus this fixed transfer condition provides no contaminated generated candidate on which to estimate incremental attack removal by the publication boundary. The boundary still executes on every valid output and recovery is checked against the original human-supported checkpoint. Five initial harness assertion failures were corrected by respecting the original restoration return type; the same provider requests were replayed from cache.

The existing synthetic candidate-suffix condition and this external payload-transfer condition have distinct denominators. The former has zero clean-rejected unsupported cases; the latter has real attempted attacks but zero observed target insertion. Their empirical contribution is the recorded interface behavior and technical attrition.

%% file: drift-results.tex
\subsection{Sequential approval dynamics}
We extend each of the same 100 SQuAD trajectories through an 8B revision and a second 14B revision. Each revision receives the question, unchanged context and preceding answer; the unchanged 14B checker assesses the new text. Human references remain scoring-only. All four stages complete for all 100 questions. The two added stages contain 333 structured invocations, 67 deterministic prechecks, and 169 additional hosted requests after cache reuse. These are continuations of the original questions, with inference reuse recorded at the exact request hash.

\begin{samepage}
Let $Y(b)\in\{0,1\}$ denote answer-reference error, $A_t$ the observed approval of candidate $b_t$, and $K_t$ the latest approved checkpoint. On the $N=65$ initially approved trajectories, with unchanged scope and acknowledged storage,
\begin{align}
K_t&=\begin{cases}b_t,&A_t=1,\\K_{t-1},&A_t=0,\end{cases}\qquad
R_t=\frac1N\sum_i Y(K_{it}),\nonumber\\
R_T-R_0&=\frac1N\sum_{i=1}^N\sum_{t=1}^T
A_{it}\bigl[Y(b_{it})-Y(K_{i,t-1})\bigr]
=\frac1N\sum_{t=1}^T(P_t-C_t).
\label{eq:telescoping-drift}
\end{align}
\end{samepage}
Here $P_t$ and $C_t$ count newly accepted error promotions and corrections relative to the immediately preceding checkpoint. Substitution of the two checkpoint cases yields each summand; summation telescopes. This identity accommodates nonmonotone revision quality and repeated acceptance of erroneous text.
\begin{table}[H]\centering\small
\begin{tabular}{lrrrrl}\toprule
Stage & Raw EM & Raw F1 & $P_t/C_t$ & Checkpoint error & $R_t-R_0$ [95\% CI]\\\midrule
Initial 3B & 72.0 & 83.4 & 0/0 & 18.5 & 0.00 [0.00, 0.00] \\
First 14B & 62.0 & 82.1 & 9/3 & 27.7 & 9.23 [-1.72, 19.61] \\
8B revision & 66.0 & 82.3 & 4/7 & 23.1 & 4.62 [-4.17, 12.90] \\
Final 14B & 69.0 & 86.2 & 7/4 & 27.7 & 9.23 [-1.72, 19.61] \\
\bottomrule\end{tabular}
\caption{Actual four-stage answer trajectories. Raw EM/F1 use all 100 questions; checkpoint error and paired differences use the 65 initially approved trajectories. Percentages and percentage-point differences use the original 42-article clustered bootstrap.}\label{tab:sequential-drift}\end{table}
The 8B stage partially reverses the first revision's error increase, and the final 14B stage reverses that recovery. The final aggregate difference equals the first revision's difference, although individual trajectories differ. All three noninitial paired intervals span zero. The observed promotion/correction counts identify the realized dynamics without imposing symmetric replacement pairs.

Question \code{5726d4a45951b619008f7f6b}, from the Victoria and Albert Museum article, asks for the Italian sculptor associated with the creation of Baroque sculpture. Its three human references identify Gian Lorenzo Bernini. The first two stages answer Bernini; 8B changes the answer to Canova and receives approval. Both names occur in the context, but they answer different entity-level propositions. This is an observed zero-token-overlap error promotion. In contrast, question \code{56d99c44dc89441400fdb5da} expands the reference ``4.7'' to ``4.7 yards per carry'', an EM error with partial token overlap. After a later rejection, the stateful policy retains that expanded checkpoint. The first case demonstrates answer-entity drift; the second demonstrates checkpoint persistence and EM boundary sensitivity.

\subsection{Evidence freshness and incomplete mediation}
The publication predicate binds answer bytes, request and contract, while the recovery predicate additionally compares an evidence fingerprint. Inspection of \code{responseActContractHash} shows goal, answer facets, response constraints and active obligations, with evidence state represented separately. We apply controlled context transitions to the 65 real initial approvals: change only request, contract or evidence; use an actual changed answer when one exists; retain an unchanged control. Storage, guard and structure conditions remain fixed. The research full-context wrapper joins the authoritative receipt with the checkpoint's evidence fingerprint and requires its equality at publication.
\begin{table}[H]\centering\small
\begin{tabular}{lrrrr}\toprule
Transition & Trials & Publication accepts & Recovery accepts & Full-context accepts\\\midrule
Unchanged & 65 & 65 & 65 & 65 \\
Answer bytes & 35 & 0 & 0 & 0 \\
Request & 65 & 0 & 0 & 0 \\
Contract & 65 & 0 & 0 & 0 \\
Evidence fingerprint & 65 & 65 & 0 & 0 \\
\bottomrule\end{tabular}
\caption{Executed freshness mediation on observed answer/receipt pairs. Context changes are controlled interventions; accepted trials count approval applicability after the intervention. The full-context column is a research wrapper around the unchanged predicate.}\label{tab:freshness}\end{table}
All 65 evidence-only transitions pass publication and fail recovery. These counts expose a predicate-level asymmetry: the publication interface cannot distinguish approval under old evidence from approval under current evidence when its observed tuple is unchanged. Source withdrawal, correction or replacement supplies the corresponding application-level freshness transition. The experiment executes this transition directly rather than attributing it to an observed attacker session.

\begin{proposition}[Context indistinguishability]
Let a publication rule observe only $\phi(b,r,c,e)=(\Hash(b),r,c)$. If an authorization requirement differs between $(b,r,c,e_0)$ and $(b,r,c,e_1)$, every deterministic rule on $\phi$ misclassifies at least one state. A randomized rule has the same acceptance probability in both states.
\end{proposition}
\begin{proof}
The two states have identical images under $\phi$. Therefore the deterministic outputs, or randomized output distributions conditional on that image, are identical, whereas the required decisions differ.
\end{proof}
Binding an injectively encoded context version $(r,c,e)$ in an authoritative receipt restores distinguishability. A changed context then requires a new semantic assessment and receipt; otherwise a version-matching eligible checkpoint is the recoverable object. This extension closes the measured applicability gap while leaving semantic false approval as a separately measured risk.

%% file: calibration-results.tex
\subsection{Scores, selection and transfer estimands}
Calibration acts on scores and decision rules. Precision, recall and discounted gain retain their original evaluation definitions. Let $p_j$ predict the independent binary label $y_j$. We measure Brier score $\mathrm{BS}=n^{-1}\sum_j(p_j-y_j)^2$ and log loss $\mathrm{NLL}=-n^{-1}\sum_j[y_j\log p_j+(1-y_j)\log(1-p_j)]$, with probabilities clipped to $[10^{-12},1-10^{-12}]$ for numerical evaluation. Source-derived probability deciles define the accompanying ECE bins. Proper scoring losses are primary because binning choices affect calibration summaries \sources{guocalibration,nixoncalibration}.

For a nondecreasing scalar map $g$, ordering candidates lexicographically by $(g(s),s)$ preserves the original ordering by $s$. Consequently nDCG, recall and reciprocal rank remain exactly invariant when calibrated ties retain the original order. We test that invariance directly. Ranking changes are attributed separately to a fitted multichannel model or a changed selection rule.

Let $\eta_D(x)=\Pr_D(Y=1\mid X=x)$ for source $S$ and target $T$. Conditional expectation yields
\begin{equation}
\mathbb{E}_T[(g(X)-Y)^2]
=\mathbb{E}_T[(g(X)-\eta_T(X))^2]
+\mathbb{E}_T[\eta_T(X)(1-\eta_T(X))].
\label{eq:calibration-transfer}
\end{equation}
The cross term vanishes conditional on $X$. Even $g=\eta_S$ therefore incurs target excess Brier risk $\mathbb{E}_T[(\eta_S-\eta_T)^2]$ where both conditionals are defined. Source-only calibration is evaluated under this distributional distinction.

\subsection{External retrieval calibration}
ArguAna supplies 1,406 queries over 8,674 documents on Hugging Face \sources{arguanahf,beir}. The corpus and judgment revisions are pinned independently. Five judged-positive document IDs are absent from the corpus; all queries and judgments remain in the denominator and those queries have zero achievable credit. The manifest records their IDs, consistent with a previously reported corpus defect \sources{arguanaissue}. Identical-query document IDs are excluded for this counterargument retrieval task. SHA-256 ordering assigns 703 queries to fitting and 703 to parameter selection. SciFact and NFCorpus retain the untouched test sets of Appendix~\ref{sec:representations}.

The candidate set is the union of each native retriever's top 100 documents. Per-query full-corpus z scores form four features, with $\log(1+s)$ applied first to BM25. Equal total query weights fit L2-regularized logistic fusion; held-out ArguAna NLL selects $C\in\{0.01,0.1,1,10,100\}$. Source standardization and coefficients are frozen before target evaluation. Scalar isotonic maps fit each sigmoid-transformed score on the 703-query selection partition. Their fixed-link references and calibrated scores share candidate sets and tie-breaking. Positive judgments define $Y=1$; unjudged candidates define the evaluation value zero. The associated probability endpoint is judged-positive retrieval, with judgment incompleteness carried by that definition.
\begin{table}[H]\centering\small
\begin{tabular}{llrrr}\toprule
Dataset & Fusion & nDCG@10 [95\% CI] & Recall@10 & MRR@10\\\midrule
SciFact & Uniform & 0.707 [0.665, 0.748] & 0.855 & 0.663 \\
SciFact & ArguAna-fitted & 0.718 [0.677, 0.760] & 0.855 & 0.677 \\
NFCorpus & Uniform & 0.360 [0.323, 0.395] & 0.177 & 0.558 \\
NFCorpus & ArguAna-fitted & 0.357 [0.322, 0.393] & 0.175 & 0.559 \\
\bottomrule\end{tabular}
\caption{Held-out retrieval after source-fitted score fusion. Both methods rank the same candidate union. Query-bootstrap intervals contain 2,000 resamples; native single-retriever rankings remain in Table~\ref{tab:retrieval}.}\label{tab:calibrated-ranking}\end{table}
On SciFact, candidate-union recall is 99.3\%. Learned-minus-uniform differences are 0.0113 [0.0001, 0.0222] in nDCG@10, 0.0000 [-0.0150, 0.0133] in Recall@10, and 0.0142 [0.0014, 0.0280] in MRR@10. The selected source regularization is $C=0.01$ for both target evaluations.
On NFCorpus, candidate-union recall is 42.3\%. Learned-minus-uniform differences are -0.0024 [-0.0073, 0.0024] in nDCG@10, -0.0020 [-0.0062, 0.0024] in Recall@10, and 0.0008 [-0.0086, 0.0106] in MRR@10. The selected source regularization is $C=0.01$ for both target evaluations.
\begin{table}[H]\centering\small
\begin{tabular}{llrrrr}\toprule
Dataset & Retriever & Fixed BS & Isotonic BS & Fixed NLL & Isotonic NLL\\\midrule
SciFact & BM25 & 0.6023 & 0.0042 & 1.8113 & 0.0231 \\
SciFact & MiniLM & 0.6871 & 0.0041 & 2.0100 & 0.0230 \\
SciFact & Multi-QA & 0.6403 & 0.0042 & 1.8048 & 0.0293 \\
SciFact & MPNet & 0.6785 & 0.0041 & 2.0313 & 0.0237 \\
NFCorpus & BM25 & 0.4412 & 0.0450 & 1.5529 & 0.5910 \\
NFCorpus & MiniLM & 0.5888 & 0.0474 & 1.6852 & 0.5279 \\
NFCorpus & Multi-QA & 0.5697 & 0.0474 & 1.6250 & 0.5646 \\
NFCorpus & MPNet & 0.5682 & 0.0476 & 1.6388 & 0.6283 \\
\bottomrule\end{tabular}
\caption{Candidate-set probability measurement before and after ArguAna isotonic fitting, with equal query weights. Fixed denotes the prespecified sigmoid link on native z scores. Every scalar-calibrated top-10 ranking is unchanged.}\label{tab:retrieval-probability}\end{table}
SciFact uniform/fitted fusion Brier scores are 0.6729/0.0040, with NLL 1.8493/0.0194. The constant source-prior reference gives Brier 0.0044 and NLL 0.0283. Candidate counts, source bin edges, ECE, coefficients and paired per-query losses accompany the result files.
NFCorpus uniform/fitted fusion Brier scores are 0.5639/0.0460, with NLL 1.5455/0.2991. The constant source-prior reference gives Brier 0.0476 and NLL 0.2619. Candidate counts, source bin edges, ECE, coefficients and paired per-query losses accompany the result files.
\subsection{External diagnostic calibration}
The logfit-project Hugging Face releases provide Thunderbird as fitting system and BGL as held-out system \sources{thunderbirdhf,bglhf,loghub}. Four source shards and both target shards are selected by a frozen filename-hash ordering. Each contributes 128 hash-selected nonoverlapping 32-event windows, restored to temporal order. The sample contains 16384 Thunderbird events (668 anomalous) and 8192 BGL events (625 anomalous). The labels originate in alert categories projected by the dataset maintainer. Raw messages, hosts, users and textual alert labels are excluded from model input. Thunderbird lacks a severity column; its missing value maps to informational severity. BGL supplies severity explicitly.

The original debugger runs unchanged with lexical, MiniLM and the pinned 1.5B Jina IQ4\_XS representations. The exact central QA input contract, pooling and dimensions remain those in the representation appendix. A read-only research observer records the seven channel scores after native fusion, and each selected sequence and score is checked against the uninstrumented pipeline. Four-fold leave-one-source-shard-out NLL chooses logistic regularization. Source out-of-fold F1 selects a separate operating threshold. The final model refits all source events; BGL labels enter only evaluation.

Native selection uses the original DPP. Score-$k$ orders native fused scores while preserving each window's native count $k$; isotonic-$k$ adds a monotone calibration with score-preserving ties. Channel-$k$ selects the same count using source-fitted probabilities. Channel-threshold applies the source-selected operating point and reports its changed count. For selected set $S_w$ and anomaly set $A_w$, pooled precision and recall are $\sum_w|S_w\cap A_w|/\sum_w|S_w|$ and $\sum_w|S_w\cap A_w|/\sum_w|A_w|$. Window recall is the mean of $|S_w\cap A_w|/|A_w|$ over anomaly-containing windows.
\begin{table}[H]\centering\small
\begin{tabular}{llrrrrl}\toprule
Input & Rule & Selected & Hits & Precision & Recall & Window recall [95\% CI]\\\midrule
Lexical & Native DPP & 544 & 53 & 9.7 & 8.5 & 15.00 [6.25, 25.87] \\
Lexical & Score-$k$ / isotonic-$k$ & 544 & 56 & 10.3 & 9.0 & 16.52 [6.25, 29.02] \\
Lexical & Channel-$k$ & 544 & 45 & 8.3 & 7.2 & 7.36 [4.66, 11.70] \\
Lexical & Channel-threshold & 7426 & 564 & 7.6 & 90.2 & 89.89 [87.42, 92.08] \\
MiniLM & Native DPP & 544 & 53 & 9.7 & 8.5 & 15.00 [6.25, 25.87] \\
MiniLM & Score-$k$ / isotonic-$k$ & 544 & 56 & 10.3 & 9.0 & 16.52 [6.25, 29.02] \\
MiniLM & Channel-$k$ & 544 & 48 & 8.8 & 7.7 & 12.25 [5.43, 22.30] \\
MiniLM & Channel-threshold & 7895 & 606 & 7.7 & 97.0 & 97.28 [96.74, 97.83] \\
Jina 1.5B & Native DPP & 544 & 53 & 9.7 & 8.5 & 15.00 [6.25, 25.87] \\
Jina 1.5B & Score-$k$ / isotonic-$k$ & 544 & 56 & 10.3 & 9.0 & 16.52 [6.25, 29.02] \\
Jina 1.5B & Channel-$k$ & 544 & 46 & 8.5 & 7.4 & 11.70 [5.43, 21.74] \\
Jina 1.5B & Channel-threshold & 7846 & 619 & 7.9 & 99.0 & 99.05 [98.23, 99.73] \\
\bottomrule\end{tabular}
\caption{Thunderbird-to-BGL diagnostic transfer on 256 held-out windows. Precision and recall are percentages. Score-$k$ and isotonic-$k$ have identical selected IDs. Window intervals resample complete anomaly-containing windows; pooled metric intervals are archived separately.}\label{tab:calibrated-logs}\end{table}
\begin{table}[H]\centering\small
\begin{tabular}{llrrr}\toprule
Input & Probability rule & Brier & NLL & ECE\\\midrule
Lexical & Native score reference & 0.1112 & 1.0896 & 0.1241 \\
Lexical & Source isotonic & 0.0717 & 0.2823 & 0.0355 \\
Lexical & Source channels & 0.0715 & 0.2820 & 0.0256 \\
MiniLM & Native score reference & 0.1112 & 1.0896 & 0.1241 \\
MiniLM & Source isotonic & 0.0717 & 0.2824 & 0.0355 \\
MiniLM & Source channels & 0.0665 & 0.2488 & 0.0857 \\
Jina 1.5B & Native score reference & 0.1112 & 1.0896 & 0.1241 \\
Jina 1.5B & Source isotonic & 0.0717 & 0.2823 & 0.0355 \\
Jina 1.5B & Source channels & 0.0688 & 0.2684 & 0.0885 \\
\bottomrule\end{tabular}
\caption{Probability transfer on all 8,192 BGL events. Native fused values are a heuristic score reference; isotonic and channel probabilities are fitted on external labels. All conditions use source-derived reliability bins.}\label{tab:log-probability}\end{table}
Lexical selects $C=0.01$ and threshold 0.02522; channel-$k$ minus native window recall is -7.65 [-18.83, 0.00] percentage points.
MiniLM selects $C=0.01$ and threshold 0.02610; channel-$k$ minus native window recall is -2.76 [-8.00, 0.00] percentage points.
Jina 1.5B selects $C=0.01$ and threshold 0.02781; channel-$k$ minus native window recall is -3.30 [-9.08, 0.00] percentage points.
The constant Thunderbird anomaly prior is 0.04077; on BGL it gives Brier 0.0717 and NLL 0.2826. The differing source and target metadata schemas create a concrete representation shift. Cross-system transfer is therefore compared against both native selection and this source-prior probability reference.
Jina channel fitting reduces probability loss but retrieves 46 labelled anomalies at the native count, compared with 53 before fitting. Its threshold rule reaches 99.0\% recall by selecting 7846 of 8192 events (95.8\%). The observed recall increase is thus accompanied by near-exhaustive selection. Proper probability scoring, sparse diagnostic utility and operating-point coverage are empirically distinct outcomes.
\subsection{Observability and representational drift}
Let $M$ be the retained metadata projection and $E=f(M)$ its embedding. Every classifier restricted to $M$ assigns a single label to events with identical metadata. In the recorded finite sample, its minimum classification error is
\begin{equation}
\widehat{R}_{M}^{*}=\frac1n\sum_m\min\{n_{m,0},n_{m,1}\}.
\label{eq:metadata-floor}
\end{equation}
Each metadata group contributes the minority count under its optimal constant decision, which proves the expression. Deterministic embeddings can merge groups but cannot separate identical inputs. For finite variables the corresponding data-processing inequality is $I(Y;E)\le I(Y;M)$. Temporal channels add information beyond this single-event projection; Equation~\ref{eq:metadata-floor} characterizes the narrower representation-only interface.
Thunderbird has 68 distinct projected metadata strings, 1 with both labels, and an empirical representation-only error floor of 4.1\%.
BGL has 8 distinct projected metadata strings, 2 with both labels, and an empirical representation-only error floor of 6.9\%.
These results connect diagnostic visibility to the security model. A source update absent from the publication context is unobservable to its authorization predicate; an event distinction removed before embedding is unobservable to the embedding function. Calibration estimates label frequencies over what remains observable. Its effectiveness under changed evidence and changed systems is measured separately from the completeness of the underlying state representation.

%% file: representation-results.tex
\section{Representation and Diagnostic Experiments}

\label{sec:representations}

\subsection{Embedding computation and retrieval judgments}

Retrieval and diagnostic representations are evaluated separately from publication authority. We run three pinned Hugging Face models: \code{all-MiniLM-L6-v2} (384 dimensions), \code{multi-qa-MiniLM-L6-cos-v1} (384), and \code{all-mpnet-base-v2} (768) \sources{minilmcard,multiqacard,mpnetcard}. All use native Transformers 4.57.6, PyTorch 2.12.0, float32, model evaluation mode, seed 1729 and local Apple MPS. Dataset and weight-file SHA-256 hashes, full repository revisions, vector hashes, reference vectors and ranked document IDs are recorded.

For hidden states $h_1,\ldots,h_T$ and attention mask $m_t$, the representation is $v=(\sum_t m_th_t)/(\sum_t m_t)$, then $e=v/\|v\|_2$. Valid special tokens participate in the mean. Raw queries and document title--newline--body strings are encoded independently with a shared 256-token envelope and batch size 32. Exact cosine similarity $e_q^\top e_d$ ranks the complete corpus; ties follow lexicographic document ID. A BM25 reference uses the full document text, Unicode alphanumeric lowercase tokens, $k_1=1.2$, $b=0.75$.

Official BEIR test judgments on Hugging Face supply 300 SciFact queries over 5,183 documents and 323 NFCorpus queries over 3,633 documents \sources{beir,scifacthf,nfcorpushf}. We report mean nDCG@10 with linear relevance gains, Recall@10 and MRR@10. Intervals use 2,000 query bootstrap resamples, with shared resamples for model differences.

\begin{table}[H]\centering\small\begin{tabular}{@{}llrrr@{}}\toprule Dataset & Retriever & nDCG@10 [95\% CI] & Recall@10 & MRR@10 \\ \midrule

scifact & MiniLM & 64.5 [59.9, 69.1] & 78.3 & 60.5 \\

scifact & MPNet & 64.1 [59.6, 68.6] & 77.7 & 60.2 \\

scifact & BM25 & 66.3 [61.6, 70.8] & 78.8 & 63.1 \\

scifact & Multi-QA MiniLM & 54.6 [49.7, 59.3] & 64.9 & 51.9 \\

nfcorpus & MiniLM & 31.7 [28.2, 35.1] & 15.5 & 50.8 \\

nfcorpus & MPNet & 33.5 [30.0, 36.9] & 16.5 & 52.6 \\

nfcorpus & BM25 & 30.6 [27.2, 34.2] & 14.9 & 51.4 \\

nfcorpus & Multi-QA MiniLM & 29.2 [25.9, 32.6] & 13.9 & 48.6 \\

\bottomrule\end{tabular}\caption{Independent BEIR relevance judgments, metrics multiplied by 100. Dense models share a 256-token encoding envelope; BM25 indexes full documents. Every test query ranks the complete corpus.}\label{tab:retrieval}\end{table}

BM25 has the highest point nDCG on SciFact, while MPNet has the highest on NFCorpus. The model ordering is dataset-dependent. Per-query paired differences, truncation counts and all configuration hashes accompany the result; retrieval scores describe ranking of judged evidence candidates.

\subsection{Deployment Jina representation in metadata-only diagnosis}

Lightcap uses \code{jina-code-embeddings-1.5b} for structural-context and run-event ranking \sources{jinacode,jinacard}. We copy the unchanged semantic debugger, structural ranker, temporal attribution and embedding module at the pinned Lightcap commit. Only imports and the embedding transport are adapted. The QA query prefix requests the most relevant answer to a question; the passage prefix identifies a candidate answer. The model-facing strings are generated by the original central contract.

The diagnostic experiment uses the production IQ4\_XS GGUF, pinned at revision \code{2330b9417fd0} \sources{jinaggufcard}; the artifact records and verifies the full revision and file SHA-256. Existing llama-cpp-python 0.3.16 produces last-token embeddings. The original 1,536-vector is truncated to its first 512 components and L2-normalized, matching the application adapter. All 130 unique model-facing strings fit within 256 tokens; the maximum observed length is 20. Encoding took 5.85 seconds on the recorded local runtime.

External labels come from all 2,000 records of Loghub\textquotesingle s BGL structured sample, containing 143 anomaly-labelled records \sources{loghub}. The original record order is split into 63 consecutive windows of at most 32 records. Component, event ID, severity and timestamp enter the run-event schema; anomaly labels, raw log messages, templates, hosts and addresses stay outside model input. Duration fields are absent. A fixed operational-failure query is issued for every window. The same copied pipeline runs with lexical ranking, MiniLM vectors, or deployment Jina vectors. The 31-event lexical frontier, 12-node MMR selection, temporal attribution, entropy weighting and DPP selection remain unchanged.

\begin{table}[H]\centering\small\begin{tabular}{@{}lrrrr@{}}\toprule Representation & Selected hits & Precision & Anomaly recall & Window recall [95\% CI] \\ \midrule

MiniLM & 44/441 & 10.0 & 30.8 & 66.4 [51.2, 81.7] \\

Jina 1.5B IQ4\_XS & 44/441 & 10.0 & 30.8 & 66.4 [51.2, 81.7] \\

Lexical & 44/441 & 10.0 & 30.8 & 66.4 [51.2, 81.7] \\

\bottomrule\end{tabular}\caption{Unchanged Lightcap debugger on public BGL metadata. Precision and anomaly recall use pooled selected events; window recall averages the anomaly-containing windows. Brackets are paired window-bootstrap intervals. Labels identify anomalous events, while causal parentage remains unlabelled.}\label{tab:debugger}\end{table}

There are 21 anomaly-containing windows. Severity-first selection at the same per-window selection sizes has 46.2\% pooled recall, and uniform selection has 24.0\% expected recall. All three representations select 44 of 143 labelled anomalies: this projection yields no observed improvement in candidate recall from Jina. Mean chronology recall on positive windows is 90.5\% with lexical ranking and 94.0\% with Jina. The chronology is a variable-size diagnostic neighborhood, so that difference is reported separately from fixed-selection recall. Per-window selections, channel weights, convergence and representation differences remain in the artifact.

%% file: transport-results.tex
\section{Integrity-Preserving Evidence Transport}
\label{sec:transport}
Evidence serialization introduces a representation change between a recorded decision and its later audit. We evaluate lossless transport on all 80 archived result files (83,585,803 bytes), including verdicts, replay journals, diagnostic scores and embedding features. GZIP uses DEFLATE's dictionary and Huffman coding \sources{deflateformat,gzipformat}; XZ carries LZMA2 streams \sources{xzformat}; Zstandard supplies a separate dictionary/entropy-coding baseline \sources{zstdformat}. Each file is encoded independently, preserving its logical path and exact bytes.

Three trials rotate codec order with gzip~\code{-9 -n}, xz~\code{-6 --threads=1}, and zstd~\code{-6 --threads=1}. Table~\ref{tab:transport} reports total stream sizes and median aggregate wall times on macOS~26.3/arm64, Python~3.14.2, Apple gzip~475, XZ~5.8.3 and Zstandard~1.5.7. Timings include process startup and pipes with in-memory inputs. All 720 file round trips recover identical bytes; compressed bytes are also stable across trials. XZ reduces the result footprint by 76.57\%; Zstandard encodes faster at a larger footprint. This single-corpus comparison determines the distributed encoding.

\begin{table}[H]
\centering\small
\caption{Lossless result transport. MB denotes $10^6$ bytes; times are aggregate seconds for 80 files.}
\label{tab:transport}
\begin{tabular}{lrrr}
\toprule
Encoding & Size (MB) & Encode (s) & Decode (s)\\
\midrule
GZIP/DEFLATE & 21.782 & 2.143 & 0.414\\
XZ/LZMA2 & 19.587 & 9.423 & 0.753\\
Zstandard & 21.085 & 0.902 & 0.513\\
\bottomrule
\end{tabular}
\end{table}

Let $C$ and $D$ denote encoding and decoding, and $P$ the publication predicate at fixed context $\gamma$. Losslessness preserves predicate evaluation:
\begin{equation}
D(C(x))=x\quad\Longrightarrow\quad P(D(C(x)),\gamma)=P(x,\gamma).
\end{equation}
The transport manifest records canonical paths, encoded and decoded lengths, and SHA-256 digests for both representations. Relative to a trusted manifest, the decoder checks the encoded object, caps decoded length, verifies the restored digest and stages the complete set before exposing it. It rejects duplicate paths, symlinks, trailing streams and existing destinations. Fifteen regression tests exercise exact recovery and malformed-input rejection. Decoder-memory and aggregate-output budgets are independently configurable. Manifest authenticity is supplied by the distribution trust boundary; byte restoration leaves semantic approval and context freshness to their existing checks.

These results establish a whole-file reference point for integrity-preserving recovery. A subsequent comparison can measure whether independently verified blocks reduce recovery latency under partial corruption, while retaining the same content-bound publication predicate and authenticated context.

%% file: references.bib
@misc{rag,
 author={Patrick Lewis and Ethan Perez and Aleksandra Piktus and Fabio Petroni and Vladimir Karpukhin and Naman Goyal and Heinrich K{\"u}ttler and Mike Lewis and Wen-tau Yih and Tim Rockt{\"a}schel and Sebastian Riedel and Douwe Kiela},
 title={Retrieval-Augmented Generation for Knowledge-Intensive {NLP} Tasks}, year={2020},
 howpublished={arXiv:2005.11401, version 4}, url={https://arxiv.org/abs/2005.11401v4}}

@misc{react,
 author={Shunyu Yao and Jeffrey Zhao and Dian Yu and Nan Du and Izhak Shafran and Karthik Narasimhan and Yuan Cao},
 title={{ReAct}: Synergizing Reasoning and Acting in Language Models}, year={2023},
 howpublished={ICLR 2023; arXiv:2210.03629, version 3}, url={https://arxiv.org/abs/2210.03629v3}}

@misc{toolformer,
 author={Timo Schick and Jane Dwivedi-Yu and Roberto Dess{\`i} and Roberta Raileanu and Maria Lomeli and Luke Zettlemoyer and Nicola Cancedda and Thomas Scialom},
 title={Toolformer: Language Models Can Teach Themselves to Use Tools}, year={2023},
 howpublished={arXiv:2302.04761}, url={https://arxiv.org/abs/2302.04761}}

@misc{indirect,
 author={Kai Greshake and Sahar Abdelnabi and Shailesh Mishra and Christoph Endres and Thorsten Holz and Mario Fritz},
 title={Not What You've Signed Up For: Compromising Real-World {LLM}-Integrated Applications with Indirect Prompt Injection}, year={2023},
 howpublished={arXiv:2302.12173}, url={https://arxiv.org/abs/2302.12173}}

@misc{injecagent,
 author={Qiusi Zhan and Zhixiang Liang and Zifan Ying and Daniel Kang},
 title={{InjecAgent}: Benchmarking Indirect Prompt Injections in Tool-Integrated Large Language Model Agents}, year={2024},
 howpublished={arXiv:2403.02691}, url={https://arxiv.org/abs/2403.02691}}

@misc{agentdojo,
 author={Edoardo Debenedetti and Jie Zhang and Mislav Balunovi{\'c} and Luca Beurer-Kellner and Marc Fischer and Florian Tram{\`e}r},
 title={{AgentDojo}: A Dynamic Environment to Evaluate Prompt Injection Attacks and Defenses for {LLM} Agents}, year={2024},
 howpublished={arXiv:2406.13352, version 3}, url={https://arxiv.org/abs/2406.13352v3}}

@misc{camel,
 author={Edoardo Debenedetti and Ilia Shumailov and Tianqi Fan and Jamie Hayes and Nicholas Carlini and Daniel Fabian and Christoph Kern and Chongyang Shi and Andreas Terzis and Florian Tram{\`e}r},
 title={Defeating Prompt Injections by Design}, year={2025},
 howpublished={arXiv:2503.18813, version 2}, url={https://arxiv.org/abs/2503.18813v2}}

@misc{fides,
 author={Manuel Costa and Boris K{\"o}pf and Aashish Kolluri and Andrew Paverd and Mark Russinovich and Ahmed Salem and Shruti Tople and Lukas Wutschitz and Santiago Zanella-B{\'e}guelin},
 title={Securing {AI} Agents with Information-Flow Control}, year={2025},
 howpublished={arXiv:2505.23643, version 2}, url={https://arxiv.org/abs/2505.23643v2}}

@misc{agentspec,
 author={Haoyu Wang and Christopher M. Poskitt and Jun Sun},
 title={{AgentSpec}: Customizable Runtime Enforcement for Safe and Reliable {LLM} Agents}, year={2025},
 howpublished={arXiv:2503.18666, version 3; accepted to ICSE 2026}, url={https://arxiv.org/abs/2503.18666v3}}

@misc{toctou,
 author={Derek Lilienthal and Sanghyun Hong},
 title={Mind the Gap: Time-of-Check to Time-of-Use Vulnerabilities in {LLM}-Enabled Agents}, year={2025},
 howpublished={arXiv:2508.17155}, url={https://arxiv.org/abs/2508.17155}}

@misc{pcaa,
 author={Zexun Wang},
 title={Proof-Carrying Agent Actions: Model-Agnostic Runtime Governance for Heterogeneous Agent Systems}, year={2026},
 howpublished={arXiv:2606.04104, version 1}, url={https://arxiv.org/abs/2606.04104v1}}

@misc{cava,
 author={Zexun Wang},
 title={{CAVA}: Canonical Action Verification and Attestation for Runtime Governance of Agentic {AI} Systems}, year={2026},
 howpublished={arXiv:2607.13716, version 1}, url={https://arxiv.org/abs/2607.13716v1}}

@misc{spa,
 author={Dylan Girrens and Guangjing Wang},
 title={{SPA}: Securing Persistent {LLM} Agents Across Queries with Plan-First Information-Flow Control}, year={2026},
 howpublished={arXiv:2608.27234, version 2}, url={https://arxiv.org/abs/2608.27234v2}}

@misc{attacker,
 author={Milad Nasr and Nicholas Carlini and Chawin Sitawarin and Sander V. Schulhoff and Jamie Hayes and Michael Ilie and Juliette Pluto and Shuang Song and Harsh Chaudhari and Ilia Shumailov and Abhradeep Thakurta and Kai Yuanqing Xiao and Andreas Terzis and Florian Tram{\`e}r},
 title={The Attacker Moves Second: Stronger Adaptive Attacks Bypass Defenses Against {LLM} Jailbreaks and Prompt Injections}, year={2025},
 howpublished={arXiv:2510.09023, version 1}, url={https://arxiv.org/abs/2510.09023v1}}

@misc{firewalls,
 author={Rishika Bhagwatkar and Kevin Kasa and Abhay Puri and Gabriel Huang and Irina Rish and Graham W. Taylor and Krishnamurthy Dj Dvijotham and Alexandre Lacoste},
 title={Indirect Prompt Injections: Are Firewalls All You Need, or Stronger Benchmarks?}, year={2026},
 howpublished={arXiv:2510.05244, version 2}, url={https://arxiv.org/abs/2510.05244v2}}

@article{saltzer,
 author={Jerome H. Saltzer and Michael D. Schroeder},
 title={The Protection of Information in Computer Systems}, year={1975},
 journal={Proceedings of the IEEE}, volume={63}, number={9}, pages={1278--1308},
 url={https://www.cs.virginia.edu/~evans/cs551/saltzer/}}

@article{robust,
 author={Andrew C. Myers and Andrei Sabelfeld and Steve Zdancewic},
 title={Enforcing Robust Declassification and Qualified Robustness}, year={2006},
 journal={Journal of Computer Security}, volume={14}, number={2}, pages={157--196},
 url={https://www.cs.cornell.edu/andru/papers/robdecl-jcs/}}

@article{entropy,
 author={Sebastian Farquhar and Jannik Kossen and Lorenz Kuhn and Yarin Gal},
 title={Detecting Hallucinations in Large Language Models Using Semantic Entropy},
 journal={Nature}, volume={630}, pages={625--630}, year={2024},
 doi={10.1038/s41586-024-07421-0}, url={https://www.nature.com/articles/s41586-024-07421-0}}

@inproceedings{ragtruth,
  title={RAGTruth: A Hallucination Corpus for Developing Trustworthy Retrieval-Augmented Language Models},
  author={Niu, Cheng and Wu, Yuanhao and Zhu, Juno and Xu, Siliang and Shum, KaShun and Zhong, Randy and Song, Juntong and Zhang, Tong},
  booktitle={Proceedings of the 62nd Annual Meeting of the Association for Computational Linguistics (Volume 1: Long Papers)},
  year={2024},
  pages={10862--10878},
  doi={10.18653/v1/2024.acl-long.585},
  url={https://aclanthology.org/2024.acl-long.585/}
}

@misc{ministralthree,
  author={{Mistral AI}}, title={Ministral 3 3B: Model Documentation},
  year={2025}, note={Dated model ministral-3b-2512; documentation inspected September 9, 2026},
  url={https://docs.mistral.ai/models/ministral-3-3b-25-12}
}

@misc{bipia,
 author={Jingwei Yi and Yueqi Xie and Bin Zhu and Emre Kiciman and Guangzhong Sun and Xing Xie and Fangzhao Wu},
 title={Benchmarking and Defending Against Indirect Prompt Injection Attacks on Large Language Models},
 year={2025}, howpublished={KDD 2025; arXiv:2312.14197v4}, doi={10.1145/3690624.3709179}, url={https://arxiv.org/abs/2312.14197v4}}

@misc{beir,
 author={Nandan Thakur and Nils Reimers and Andreas R{\"u}ckl{\'e} and Abhishek Srivastava and Iryna Gurevych},
 title={{BEIR}: A Heterogenous Benchmark for Zero-shot Evaluation of Information Retrieval Models},
 year={2021}, howpublished={NeurIPS Datasets and Benchmarks; arXiv:2104.08663v4}, url={https://arxiv.org/abs/2104.08663v4}}

@misc{minilmcard, author={{Sentence Transformers}}, title={all-MiniLM-L6-v2: Model Card}, year={2026},
 note={Revision 1110a243fdf4706b3f48f1d95db1a4f5529b4d41}, url={https://huggingface.co/sentence-transformers/all-MiniLM-L6-v2}}

@misc{multiqacard, author={{Sentence Transformers}}, title={multi-qa-MiniLM-L6-cos-v1: Model Card}, year={2026},
 note={Revision b20736733232; full pin in the artifact}, url={https://huggingface.co/sentence-transformers/multi-qa-MiniLM-L6-cos-v1}}

@misc{jinaggufcard,author={{Jina AI}},title={{jina-code-embeddings-1.5b-GGUF}: Quantized model repository},year={2026},
 note={IQ4\_XS, revision 2330b9417fd0; full pin and file digest in the artifact},url={https://huggingface.co/jinaai/jina-code-embeddings-1.5b-GGUF}}

@misc{mpnetcard, author={{Sentence Transformers}}, title={all-mpnet-base-v2: Model Card}, year={2026},
 note={Revision e8c3b32edf5434bc2275fc9bab85f82640a19130}, url={https://huggingface.co/sentence-transformers/all-mpnet-base-v2}}

@misc{scifacthf, author={{BEIR}}, title={SciFact Corpus, Queries and Relevance Judgments}, year={2026},
 note={Pinned Hugging Face dataset and qrels revisions in the artifact}, url={https://huggingface.co/datasets/BeIR/scifact}}

@misc{nfcorpushf, author={{BEIR}}, title={NFCorpus Corpus, Queries and Relevance Judgments}, year={2026},
 note={Pinned Hugging Face dataset and qrels revisions in the artifact}, url={https://huggingface.co/datasets/BeIR/nfcorpus}}

@misc{jinacode,
 author={Daria Kryvosheieva and Saba Sturua and Michael G{\"u}nther and Scott Martens and Han Xiao},
 title={Efficient Code Embeddings from Code Generation Models}, year={2025}, howpublished={arXiv:2508.21290}, url={https://arxiv.org/abs/2508.21290}}

@misc{jinacard,author={{Jina AI}},title={jina-code-embeddings-1.5b: Model Card and Task Configuration},year={2026},
 note={Revision 39aeb4fb9b60f930934c78ae5d749a46287c248a},url={https://huggingface.co/jinaai/jina-code-embeddings-1.5b}}

@misc{loghub,
 author={Jieming Zhu and Shilin He and Pinjia He and Jinyang Liu and Michael R. Lyu},
 title={Loghub: A Large Collection of System Log Datasets for {AI}-driven Log Analytics},year={2023},
 howpublished={ISSRE 2023; arXiv:2008.06448v3},url={https://arxiv.org/abs/2008.06448v3}}

@misc{squad,
 author={Pranav Rajpurkar and Jian Zhang and Konstantin Lopyrev and Percy Liang},
 title={{SQuAD}: 100,000+ Questions for Machine Comprehension of Text},year={2016},
 howpublished={EMNLP 2016; arXiv:1606.05250},url={https://arxiv.org/abs/1606.05250}}

@misc{squadhf,author={Pranav Rajpurkar},title={{SQuAD}: Hugging Face Dataset},year={2026},
 note={Revision 7b6d24c440a36b6815f21b70d25016731768db1f},url={https://huggingface.co/datasets/rajpurkar/squad}}

@misc{ministraleight,
  author={{Mistral AI}}, title={Ministral 3 8B: Model Documentation},
  year={2025}, note={Dated model ministral-8b-2512; documentation inspected September 9, 2026},
  url={https://docs.mistral.ai/models/ministral-3-8b-25-12}
}

@inproceedings{guocalibration,
 author={Chuan Guo and Geoff Pleiss and Yu Sun and Kilian Q. Weinberger},
 title={On Calibration of Modern Neural Networks},booktitle={Proceedings of the 34th International Conference on Machine Learning},
 series={Proceedings of Machine Learning Research},volume={70},pages={1321--1330},year={2017},url={https://proceedings.mlr.press/v70/guo17a.html}}

@inproceedings{nixoncalibration,
 author={Jeremy Nixon and Michael W. Dusenberry and Linchuan Zhang and Ghassen Jerfel and Dustin Tran},
 title={Measuring Calibration in Deep Learning},booktitle={IEEE/CVF Conference on Computer Vision and Pattern Recognition Workshops},
 pages={38--41},year={2019},url={https://openaccess.thecvf.com/content_CVPRW_2019/papers/Uncertainty%20and%20Robustness%20in%20Deep%20Visual%20Learning/Nixon_Measuring_Calibration_in_Deep_Learning_CVPRW_2019_paper.pdf}}

@misc{arguanahf,author={{BEIR}},title={ArguAna Corpus, Queries and Relevance Judgments},year={2026},
 note={Corpus revision 9bcf8fc0320c; qrels revision ae5468c6f1c1; full pins in the artifact},url={https://huggingface.co/datasets/BeIR/arguana}}

@misc{arguanaissue,author={{BEIR contributors}},title={Missing document from corpus ArguAna, issue 172},year={2024},url={https://github.com/beir-cellar/beir/issues/172}}

@misc{thunderbirdhf,author={{logfit-project}},title={Thunderbird System Logs: Hugging Face Dataset},year={2026},
 note={Revision 6b4a6401591e85a534c401a49355553631133bd0},url={https://huggingface.co/datasets/logfit-project/Thunderbird}}

@misc{bglhf,author={{logfit-project}},title={BGL System Logs: Hugging Face Dataset},year={2026},
 note={Revision 09a26419fc644a2f96e96d18080397f0e212c8a6},url={https://huggingface.co/datasets/logfit-project/BGL}}

@misc{ministralfourteen,
  author={{Mistral AI}}, title={Ministral 3 14B: Model Documentation},
  year={2025}, note={Dated model ministral-14b-2512; documentation inspected September 9, 2026},
  url={https://docs.mistral.ai/models/ministral-3-14b-25-12}
}

@techreport{gzipformat,
 author={Peter Deutsch},title={{GZIP} File Format Specification Version 4.3},
 institution={RFC Editor},number={RFC 1952},year={1996},doi={10.17487/RFC1952},
 url={https://www.rfc-editor.org/rfc/rfc1952}}

@techreport{deflateformat,
 author={Peter Deutsch},title={{DEFLATE} Compressed Data Format Specification Version 1.3},
 institution={RFC Editor},number={RFC 1951},year={1996},doi={10.17487/RFC1951},
 url={https://www.rfc-editor.org/rfc/rfc1951}}

@misc{xzformat,
 author={Lasse Collin and Igor Pavlov},title={The .xz File Format},
 year={2024},note={Version 1.2.1, April 8},
 url={https://tukaani.org/xz/xz-file-format-1.2.1.txt}}

@techreport{zstdformat,
 author={Yann Collet and Murray Kucherawy},title={Zstandard Compression and the `application/zstd' Media Type},
 institution={RFC Editor},number={RFC 8878},year={2021},doi={10.17487/RFC8878},
 url={https://www.rfc-editor.org/rfc/rfc8878}}
